\documentclass{article}

\usepackage[ngerman,english]{babel}

\usepackage[utf8]{inputenc}
\usepackage{amsthm}
\usepackage{scrpage2}
\automark[section]{chapter}
\setheadsepline{0.5pt}

\newcommand{\norm}[1]{\left\|#1\right\| }
\newcommand{\betr}[1]{\left|#1\right|}
\newcommand{\epsi}[0]{\varepsilon}
\newcommand{\im}[0]{\mathrm{i}}

\newcommand{\ee}[0]{\mathrm{e}}
\newcommand{\rmm}[1]{{\mathrm{#1}}}
\usepackage{marginnote}

\usepackage{amsmath}
\usepackage{mathrsfs}
\usepackage{amssymb}
\usepackage{bbm}
\usepackage{color}
\usepackage{graphicx}
\usepackage{float}

\newcommand{\field}[1]{\mathbb{#1}}

\newcommand{\R}{\field{R}}
\newcommand{\N}{\field{N}}
\newcommand{\id}[0]{\mathbf{1}}
\usepackage{verbatim}

\newcommand{\Or}{{\mathcal{O}}}
 \newcommand{\E}{{\mathrm{e}}}
\newcommand{\I}{\mathrm{i}}
\newcommand{\D}{{\mathrm{d}}}

\newcommand{\fock}{ \mathcal{F} }
\DeclareMathOperator{\ran}{ran}

\theoremstyle{plain}
\newtheorem{thm}{Theorem}
\newtheorem{lem}[thm]{Lemma}
\newtheorem{prop}[thm]{Proposition}
\newtheorem{cor}[thm]{Corollary}
\newtheorem*{thm*}{Theorem}%
\newtheorem*{lem*}{Lemma}
\newtheorem*{prop*}{Proposition}
\newtheorem*{asum}{Assumptions} 
\newtheorem*{cor*}{Corollary}

\theoremstyle{definition}
\theoremstyle{remark}

\title{Non-adiabatic transitions in a massless scalar field}
\author{Johannes von Keler  and Stefan Teufel\\[1mm]
\em \small Mathematisches Institut, Universit\"at T\"ubingen, Germany}
 
\begin{document}
\maketitle

\begin{abstract}
We consider the dynamics of a massless scalar field with time-dependent sources in the adiabatic limit. 
This is an example of an adiabatic problem without spectral gap.
The main goal of our paper is to illustrate the difference between the error of the adiabatic approximation 
and the concept of non-adiabatic transitions for gapless systems. In our example the non-adiabatic transitions
correspond to emission of free bosons, while the error of the adiabatic approximation is dominated by a velocity-dependent deformation of the ground state of the field. 
In order to capture these concepts precisely, we show how to construct super-adiabatic approximations for a gapless system.
\end{abstract}

\section{Introduction}

The adiabatic theorem of quantum mechanics is usually stated and proved for systems with a finite spectral gap. Mathematically it is known since at least 1998~\cite{Bo,AvEl99} that in a weaker sense the theorem remains valid even for systems without spectral gap as long as the spectral projections of the time-dependent Hamiltonian are sufficiently regular functions of time, see also~\cite{teufel2}. Its validity has been shown even for  resonances~\cite{AF}  and for open systems~\cite{AFGG}. In recent years the problem of adiabaticity for gapless systems attracted  also interest in physics. For example, in~\cite{PoGr} the authors consider adiabatic changes in the coupling of a scalar field. In a finite volume the ground state of the field is separated by a gap from the excited states. However, in the thermodynamic limit the gap closes and the spectrum of the field becomes absolutely continuous. The authors find that in the thermodynamic limit the error in the adiabatic approximation depends on the spatial dimension and other parameters controlling the spectral density of the field. They distinguish three regimes: one where the adiabatic theorem holds with the same error estimate as in the gapped case, one where it holds with a different power law and one where it fails completely. In all cases, however, they identify the error of the adiabatic approximation with the size of the non-adiabatic transitions. One main goal of our paper is to explain why this is, in general, not a valid identification and why  the concept of super-adiabatic approximations is useful also  in the gapless case.

To this end we consider a specific gapless model in the adiabatic limit, namely a massless scalar field with   time dependent sources. The goal is to exhibit a number of subtleties in  adiabatic theory that are known in the gapped case -- but often not expressed sufficiently clearly -- also for the case of gapless systems. In particular we emphasize the difference between the error in the adiabatic approximation and what physically should be considered non-adiabatic transitions. To be more specific, let
us briefly recall the situation in the presence of a spectral gap. Let  $H(t)$ be a time-dependent Hamiltonian and $E(t)$ an eigenvalue with spectral projection $P(t)$. Under appropriate regularity conditions on $H(t)$ and the assumption that $E(t)$ is separated by a gap from the rest of the spectrum of $H(t)$, the adiabatic theorem states that the solution of the Schr\"odinger equation
\[
\I \epsi \tfrac{\D}{\D t} \psi(t)  = H(t) \psi(t) \quad\mbox{ with } \quad \psi(0)\in \ran P(0)
\]
satisfies
\begin{equation}\label{IA1}
\| (1-P(t) )\psi(t) \| = \Or(\epsi)\,,
\end{equation}
i.e.\ that any solution starting in the spectral subspace $ \ran P(0)$ evolves into the subspace $\ran P(t)$ up to an error of order $\epsi$.
This error estimate, which contributes to the error of the adiabatic approximation, is optimal in the sense that generically  the piece $(1-P(t) )\psi(t)$ of the solution in the orthogonal complement of $\ran P(t)$ is
really of order $\epsi$ and not smaller   for any finite time  $t$ where $\frac{\D}{\D t}H(t)\not=0$. However, if $H(t)$ is constant outside of a  bounded interval, say $[0,T]$,
and $N+1$-times continuously differentiable, then clearly $\psi(t)\in  \ran P(t)$ for $t\leq 0$, but also
\begin{equation}\label{IA2}
\| (1-P(t) )\psi(t) \| = \Or(\epsi^N)\quad\mbox{ for }\quad t\geq T\,.
\end{equation}
Thus the non-adiabatic transitions into the orthogonal complement of $\ran P(t)$ after a compactly supported adiabatic change of the Hamiltonian are much smaller than~$\Or(\epsi)$. 
What happens is that during the adiabatic change of the Hamiltonian the solution does not exactly follow the spectral subspace  $\ran P(t)$, which we will call the adiabatic subspace in the following,
but slightly tilted $\epsi$-dependent subspaces ran$P^\epsi(t)$, so-called superadiabatic subspaces. With respect to the superadiabatic subspaces one has 
\begin{equation}\label{IA3}
\| (1-P^\epsi(t) )\psi(t) \| = \Or(\epsi^N)\quad\mbox{ for {\bf all} }  t \in \R\,.
\end{equation}
In addition it holds that 
\[
\| P^\epsi(t) - P(t) \| = \Or(\epsi)
\]
which together with (\ref{IA3}) implies the usual adiabatic estimate (\ref{IA1}). 
Moreover, at times~$t$ when all derivatives of $H(t)$ vanish, the superadiabatic projections agree with the adiabatic projection, $P^\epsi(t) = P(t)$, and thus  (\ref{IA3}) implies also~(\ref{IA2}). In Figure~\ref{transhist} the situation is plotted schematically. 

\begin{figure}[h]\begin{center}\vspace{1cm}
\includegraphics[height=4cm]{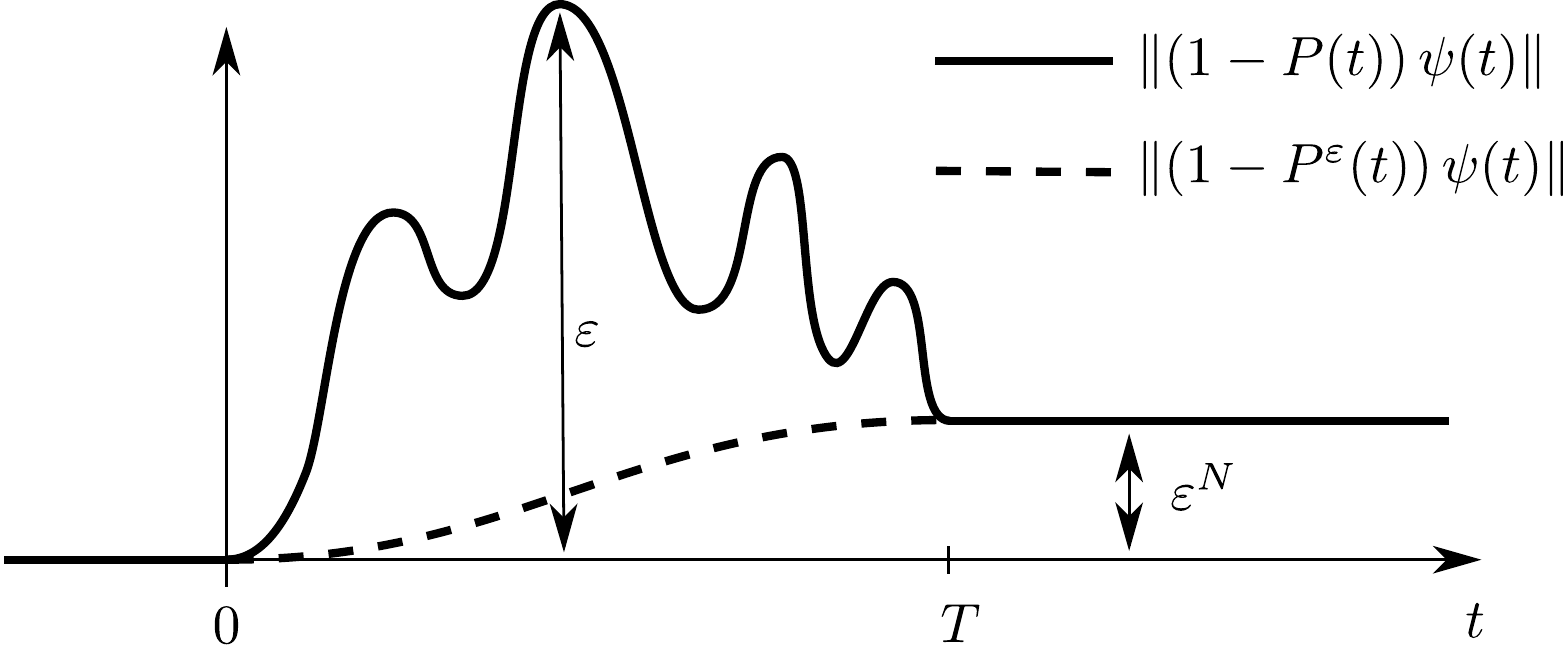}
\caption{\small Schematic plot of the transition histories with respect to the adiabatic (full line) and the superadiabatic (dashed line) subspaces. The Hamiltonian varies only during the time interval $[0,T]$ and outside of this interval the adiabatic subspace $P(t)$, i.e.\ the instantaneous spectral subspace, agrees with the superadiabatic subspace $P^\epsi(t)$. While the non-adiabatic transitions for $t\geq T$ are of order $\epsi^N$, the norm of the piece of the solution that leaks into the orthogonal complement of $P(t)$ for $0<t<T$ is typically much larger, namely of order $\epsi$. The norm of the piece of the solution that leaks into the orthogonal complement of the superadiabatic subspace $P^\epsi(t)$ for $0<t<T$ remains small, i.e.\ of order $\epsi^N$, for all times.\label{transhist}} 
\end{center}
\end{figure}

In this sense we say that the error of the adiabatic approximation is of order $\epsi$, while the non-adiabatic transitions are (at most) of order $\epsi^N$. For a  detailed exposition of adiabatic theory with spectral gap we refer to~\cite{teufel1} and references therein. Here let us mention only that the notion ``superadiabatic'' was termed by M.\ Berry in~\cite{Be}. There he shows how for Hamiltonians depending  analytically on time the transitions between superadiabatic subspaces are even exponentially small in $\epsi$ (as it is known also from the famous Landau-Zener model) and moreover, that the transitions as a function of time follow the universal shape of an error function, i.e.\ an integrated Gaussian. A rigorous account of~\cite{Be} is given in~\cite{BeTe1,BeTe2,BeTe3}. 

In the present paper we show within a physically relevant but relatively simple example that   in the gapless case the situation is to some extend similar. Before presenting our results in detail let us introduce the model and explain some of its important features. 

We consider a scalar massless field in three spatial dimensions. The momentum space for a single boson is $L^2(\R^3_k)$ and the state space of the field is the symmetric Fock space
\[
\fock :=\textstyle  \bigoplus_{n=0}^\infty L^2(\R^3)^{\otimes_s n}\,.
\]
So $\psi\in\fock$ is a sequence $\psi=(\psi_0, \psi_1,\psi_2,\ldots)$ with $\psi_n(k_1,\ldots,k_n)$ a square integrable symmetric function of $n$ variables in $\R^3$. For $\psi,\phi\in\fock$ the inner product is
\[
\langle \psi,\phi\rangle_\fock = \textstyle\sum_{n=0}^\infty \langle \psi_n,\phi_n\rangle_{L^2(\R^{3n})}\,.
\]
The Hamiltonian of the free field is $H_{\rm f} := \D\Gamma(|k|)$ and acts as
\[
(H_{\rm f} \psi)_n(k_1, \ldots, k_n)  =\textstyle \sum_{i=1}^n |k_i| \psi_n(k_1,\ldots, k_n)\,.
\]
As a multiplication operator $H_{\rm f}$ is self-adjoint on 
  its dense maximal  domain~$D(H_{\rm f})$. Moreover, $H_{\rm f}$ has a unique ground state given by the Fock vacuum $\Omega_0 := (1,0,0,\ldots)$. Now we add  moving charges as sources   to 
the field. For notational simplicity we assume that all sources have the same normalized form factor $\varphi:\R^3\to [0,\infty)$ with
\[
\textstyle \int_{\R^3} \D x\, \varphi(x)  =1\,.
\]
Later on we will make additional assumptions on the ``charge distribution''~$\varphi$. The sources are   located at positions $x_j \in\R^3$  with total charge $e_j$ and form factor $\varphi$. The linear coupling to the field is given by the  operator 
\[
H_{\rmm I}(x)  := \textstyle \sum_{j=1}^N H_{{\rmm I},j}(x_j) := \sum_{j=1}^Ne_j \,\Phi\Big( \frac{ \hat \varphi(k)}{\sqrt{|k|} }\,\E^{\I k\cdot x_j }\Big)\,,
\]
 where we abbreviate $x=(x_1,\ldots, x_N)$.
Here $\hat\varphi$ denotes the Fourier transform of $\varphi$ and~$\Phi$ is the field operator
\[
\Phi(f) := \tfrac{1}{\sqrt{2}}\big( a^\dagger(f) + a(  f) \big)\,,
\]
where $a$ and $a^\dagger$ are the standard bosonic annihilation and creation operators on~$\fock$, see Section~\ref{sec:fock} for more details.
Under appropriate conditions on $\varphi$  the total Hamiltonian
$H(x) = H_{\rm f} + H_{\rmm I}(x)$
is self-adjoint on $D(H):= D(H_{\rm f})$ and bounded from below. If the total charge of the sources $e=\sum_{j=1}^N e_j$ is zero, $H(x)$  has a unique ground state $\Omega(x)\not= \Omega_0$. The ground state $\Omega(x)$ contains so-called ``virtual bosons'' that provide a ``dressing'' of the sources. 
In contrast to ``free bosons'', the virtual bosons do not propagate. These heuristic notions
will be made more precise later on. 
If the total charge is different from zero there is still a good notion of approximate ground state, at least for our purposes. 

We will be interested in the dynamics of the scalar field when the sources change their locations with time.
Let $x_j:\R\to \R^3$, $t\mapsto x_j(t)$ be smooth, then 
\[
H(t) := H(x(t)) = H_{\rm f} + \textstyle \sum_{j=1}^N H_{{\rmm I},j}(x_j(t))
\]
describes the evolution of the scalar field in the presence of sources moving along the prescribed trajectories $x_j(t)$. We consider the solutions of the Schr\"odinger equation
\[
\I\epsi\tfrac{\D}{\D t} \psi(t) = H(t)\psi(t)
\]
for the field in the adiabatic limit $\epsi\ll 1$, which corresponds to slowly moving sources. Assume for a moment that the field is restricted to a finite box. Then the Hamiltonian $H(t)$ has purely discrete spectrum and the ground state is separated from the first excited state by a gap. The adiabatic theorem implies that if the field starts in the ground state $\Omega(x(0))$ of $H(0)$, it remains up to a phase $\epsi$-close to the ground state $\Omega(x(t))$ of $H(t)$ at later times. The solution is thus $\epsi$-close to a state where a static cloud of virtual bosons follows the slowly moving sources and no radiation is emitted. The error of the adiabatic approximation is of order $\epsi$.
If we pass to the superadiabatic approximation, we realize that the solution is even $\epsi^N$-close to a state $\Omega^\epsi(t)$ where an $\epsi$-dependent cloud of virtual  bosons  follows the slowly moving sources and no radiation is emitted. The state $\Omega^\epsi(t)$ is still $\epsi$-close to the static ground state $\Omega(x(t))$, however, the dressing by virtual bosons depends now also on the velocities $\dot x_j(t)$ and higher derivatives of $x_j(t)$. If at some time the sources come to rest then from that time on $\Omega^\epsi(t)=\Omega(x(t))$, hence the field is $\epsi^N$-close to the static ground state again. In particular, the probability for emitting a free boson is at most of order~$\epsi^{2N}$. 

\begin{figure}[h]\begin{center}\vspace{1cm}
\includegraphics[height=10cm]{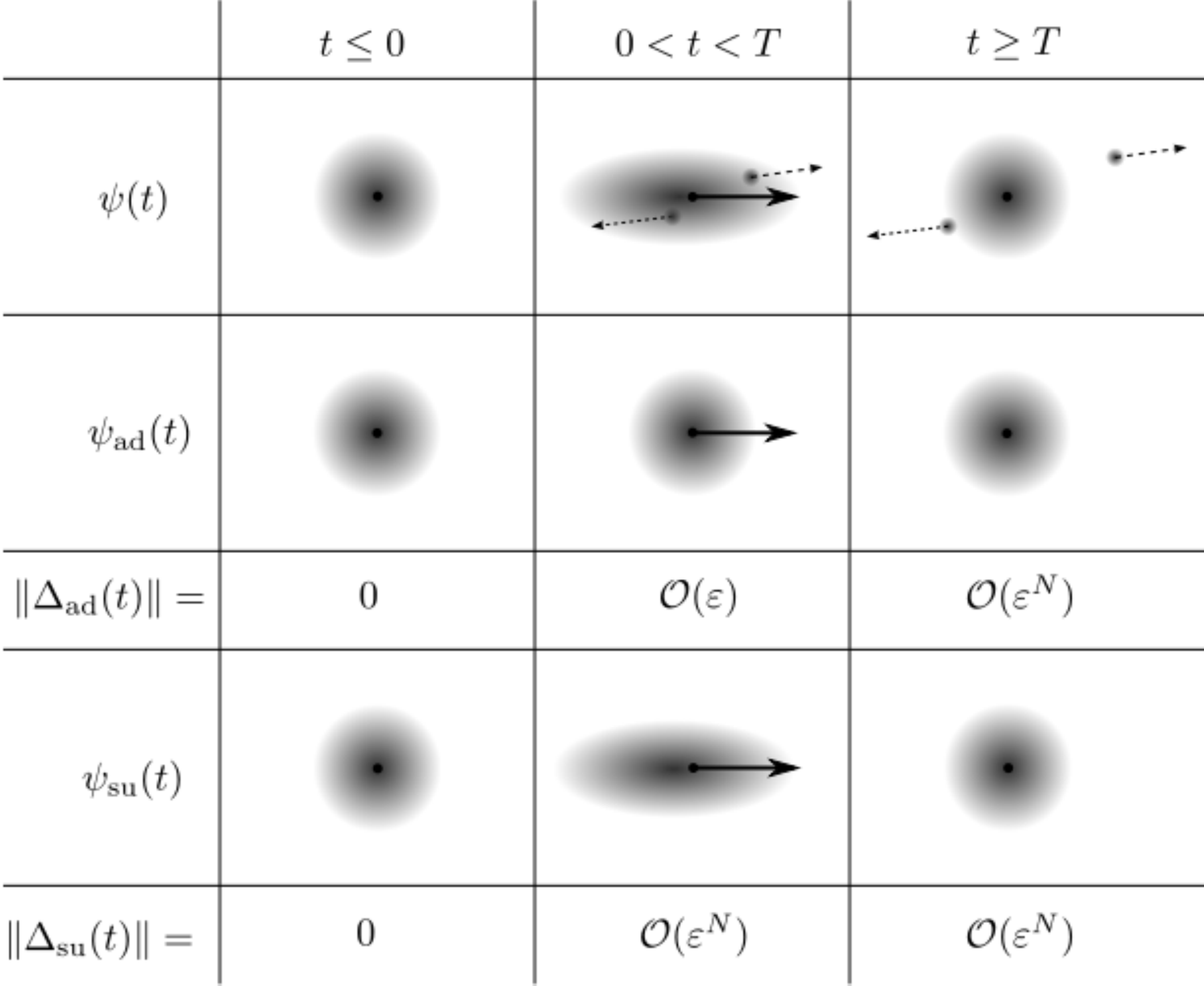}
\caption{\small The situation with spectral gap: For $t\leq 0$ the sources are at rest and the adiabatic approximation agrees with the true solution. While the sources are moving, the dressing is deformed, which leads to an error of order $\epsi$ in the adiabatic approximation. The superadiabatic approximation takes into account the deformation and the error is only of order $\epsi^N$. When the sources are again at rest for $t\geq T$, the dressing is again the static one and the error of the adiabatic and the superadiabatic approximation is  of order $\epsi^N$. This error is due to 
non-adiabatic transitions that correspond  to the emission of free bosons.\label{f2} }
\end{center}
\end{figure}

Summarizing the above, we have that
as long as there is a gap and as long as the prescribed trajectories $x_j(t)$ are $C^{N+1}(\R,\R^3)$, non-adiabatic transitions corresponding to the emission of free bosons are at most  of order~$\epsi^N$. 
Nonetheless, the error of the adiabatic approximation is of order $\epsi$ for all times where $\dot x_j(t)\not=0$ because of the deformation of the cloud of virtual bosons, see also Figure~\ref{f2}. This exemplifies in a simple physical example why it is important to distinguish carefully between the  error of the adiabatic approximation and the non-adiabatic transitions. It also is an example where superadiabatic subspaces have a clear physical meaning as velocity-dependent dressing of the sources. This is of course all well known and 
the content of this paper is to show that this picture survives also in the case without spectral gap, i.e.\ in the thermodynamic limit where the box is replaced by $\R^3$.

In the remainder of the introduction we provide informal statements of our main results, where for simplicity we restrict to the case of a neutral system.
Then there is a unitary transformation $V(x)\in\mathcal{L}(\fock)$ such that $V(x)H(x)V(x)^* = H_{\rm f} + E(x)$ and thus  $H(x)$  has a unique ground state
$\Omega(x)=V^*(x)\Omega_0$ with eigenvalue $E(x)$. The   general and rigorous statements are explained in Section~2. Our first result is an adiabatic theorem without spectral gap, cf.\ Theorem~\ref{thm:adapp}.\\[2mm]
\noindent {\bf Adiabatic Theorem:} {\em 
The solution of 
\begin{equation}\label{schroe}
\I\epsi\tfrac{\D}{\D t} \psi(t) = H( x(t) )\psi(t) 
\end{equation}
with $\psi(0) = \Omega(x(0))$ satisfies for any $t\in\R$ 
\begin{equation}\label{adi0}
\| \psi(t) - \E^{-\frac{\I}{\epsi}\int_0^t \D s\, E(x(s))} \Omega(x(t))\| = \Or \Big( \epsi \sqrt{\ln(1/\epsi)} \Big)\,.
\end{equation}
}

Thus, although the spectrum of $H(x)$ is the whole half line $[E(x),\infty)$ and the eigenvalue $E(x)$ is not separated by a gap from the continuous spectrum, the usual adiabatic approximation for the evolution of eigenstates holds with almost the same error estimate as in the gapped case.  
We have even more: if one adds free bosons then their dynamics is adiabatic too in the following sense. To each configuration $x$ of the sources the annihilation and creation operators of free bosons are $b_x := V(x)^* a V(x)$ and $b^\dagger_x:=V(x)^* a^\dagger V(x)$ and thus
\[
\E^{\frac{\I}{\epsi}H(x)t }\, b^\#_x( f) \,\E^{-\frac{\I}{\epsi}H(x)t } \;=\;
b^\#_x( f(t) )\,,
\]
where $f(k,t) = \E^{-\frac{\I}{\epsi}|k|t}f(k)$ is just the free time evolution of a boson and $b^\#_x$ stands either for $b_x$ or $b_x^\dagger$.
Now we define the adiabatic approximation   as
\[
\psi_{\rm ad}(t) :=  \E^{-\frac{\I}{\epsi}\int_0^s \D s\, E(x(s))} \prod_{l=1}^m b^\dagger_{x(t)}(f_l(t) ) \, \Omega(x(t))\,.
\]
We will show that it
approximates 
the solution of (\ref{schroe}) with initial datum 
\[
\psi(0) = \prod_{l=1}^m b^\dagger_{x(0)}(f_l) \,\Omega(x(0))\,,\quad f_l\in L^2(\R^3)\,,
\]
in the same sense as in (\ref{adi0}), 
 \begin{equation}\label{adiapprox}
\left\| \psi(t) - \psi_{\rm ad}(t) \right\| = \Or \big( \epsi\sqrt{\ln(1/\epsi)}\big)\,. 
\end{equation}

If we denote by $Q_m(x):=V(x)^* \, Q_m\,V(x)$ the projection on the sector of Fock space containing exactly $m\in\N_0$ free bosons, where $Q_m$ is the projection on the $m$-particle sector of Fock space,
the above result implies that for $\psi(0)\in Q_{m}(x(0)) \fock$
\[
\| (1-Q_{m}(x(t)))\, \psi(t) \|^2 = \Or \big( \epsi^2 \ln(1/\epsi)\big)\,,
\]
 i.e.\ that the probability for emitting a free boson is at most of order $\epsi^2  \ln(1/\epsi )$.
Note that this means that not only the spectral subspace $Q_0(x(t))\fock$ is adiabatically invariant, but also the subspaces $Q_m(x(t))\fock$, which are not spectral subspaces of the instantaneous Hamiltonian $H(x(t))$.

In order to understand what part of $(1-Q_{m}(x(t)))\, \psi(t) $ really corresponds to emission of free bosons and what part is merely a velocity-dependent deformation of the dressing, we need to introduce the superadiabatic
picture. Now not only  the dressing of the sources depends  on their velocities, but also the annihilation and creation operators of free bosons and the associated $m$-particle sectors of Fock space. 
In Lemma~\ref{lem:V} we construct a velocity-dependent dressing operator $V^\epsi_\sigma(t)$ such that the corresponding objects $\Omega^\epsi(t) := V^{\epsi*}(t)\Omega_0$, $b^{\epsi \#}_t:=V^{\epsi*} (t) a^\# V^\epsi_\sigma(t) $ and $Q_m^\epsi(t):=V^{\epsi*} (t) Q_m V^\epsi_\sigma(t) $
allow for an improved adiabatic approximation. Let
\[
\psi_{\rm su}(t) := \E^{-\frac{\I}{\epsi}\int_0^s \D s\, E^\epsi(x(s))}  \prod_{l=1}^m b^{\epsi \dagger}_t(f_l(t)  ) \, \Omega^\epsi( t )
\]
be the superadiabatic approximation to the solution $\psi(t)$ of (\ref{schroe}) with initial datum 
\[
\psi(0) = \prod_{l=1}^m b^{\epsi \dagger}_0(f_l) \,\Omega^\epsi (0)\,,\quad f_l\in L^2(\R^3)\,,
\]
then  (\ref{adiapprox}) can be improved in the following sense: the part of the solution $\psi(t)$ that stays in the $m$-particle sector is given by $\psi_{\rm su}(t)$ with a much smaller error,
\[
\| Q_m^\epsi(t) \psi(t) - \psi_{\rm su}(t) \| = \Or \big( \epsi^2 \ln(1/\epsi)\big)\,.
\]
However,   more importantly, the superadiabatic representation can be used to compute the non-adiabatic transitions corresponding to the emission of free bosons by first order perturbation theory. In Theorem~\ref{thm:foae} we show that the ``non-adiabatic''    wave function
\[
\psi_{\rm na}(t) :=   -\I\epsi 
 \int_0^t \D s\, \E^{-\frac{\I}{\epsi}   \int_s^t \D \tilde s\, E^\epsi (x(\tilde s)) }\,   \,\Phi^\epsi_{s} \left( \E^{-\frac{\I}{\epsi}|k|t} g^\epsi(s) \right)   \psi_{\rm su}(s)  
\]
satisfies
\[
\left\| (1-Q_m^\epsi(t)) \,\psi(t) - \psi_{\rm na}(t) \right\| = \Or\left( \epsi^2 \ln(1/\epsi)\right)\,.
\]
Here $\Phi^\epsi_{s} (f) := \frac{1}{\sqrt{2}} \left(b^{ \epsi \dagger}_{s} (f) + b^\epsi_{s}(  f)\right)$ and  $g^\epsi(s) = g^\epsi(  x(s), \ddot x(s))$ is a coupling function depending on the acceleration of the sources.
In summary 
\[
\|   \psi(t) - (\psi_{\rm su}(t)+ \psi_{\rm na}(t)) \| = \Or\big( \epsi^2 \ln(1/\epsi)\big) 
\]
and thus $\psi_{\rm su}(t) + \psi_{\rm na}(t)$ yields a good approximation of the true solution $\psi(t)$ with a clear separation of a ``superadiabatic'' piece $\psi_{\rm su}(t)$ that improves the adiabatic approximation and a ``non-adiabatic'' piece $\psi_{\rm na}(t)$ that contains real non-adiabatic transitions.  Note that a posteriori one can compute the leading order of $\psi_{\rm na}(t)$ replacing all superadiabatic objects by their leading order adiabatic counterparts,
\[
\psi_{\rm na}(t)  =   -\I\epsi 
 \int_0^t \D s\, \E^{-\frac{\I}{\epsi}   \int_s^t \D \tilde s\, E  (x(\tilde s)) }\,   \,\Phi_{ s } \left( \E^{-\frac{\I}{\epsi}|k|t} g^\epsi(s) \right)   \psi_{\rm ad}(s)  \;+\; \Or \big( \epsi^2 \ln(1/\epsi)\big)\,.
\]

Finally one can ask for the probability of emitting a free boson, i.e.\ for computing $\|\psi_{\rm na}(t)\|^2$. However, due to the infrared problem it turns out that this has no nice asymptotics for $\epsi\to 0$ and
we can only show 
$\|\psi_{\rm na}(t)\|^2 =\Or \big( |t|^2 \epsi^2 \ln(1/\epsi) \big)$. But a physically more interesting question is the energy radiated through the emission of free bosons. In Theorem~\ref{thm:rad} we show that when starting in the dressed vacuum, i.e.\  $\psi(0) = \Omega^\epsi(0)$, then
\[
E_{\rm rad}(t) \;:=\;  
\langle \psi_{\rm na}(t), H_{\rm f}(t)  \, \psi_{\rm na}(t) \rangle_\fock \; =\; 
\frac{\epsi^3}{2\cdot 6 \pi}\int_0^t \D s\, | \ddot d (s)|^2 \;+ \;{o}(\epsi^3)\,, 
\]
 where 
 \[
 \ddot d(t) :=\textstyle \sum_{j=1}^N e_j \ddot x_j(t) 
 \]
 is the second derivative of the dipole moment of the sources and $H_{\rm f}(t)$   is the field energy of the free photons. Taking a derivative gives the Larmor formula for the radiated power of slowly moving charges,
 \[
 P_{\rm rad}(t) := \frac{\D}{\D t} E_{\rm rad}(t) =   \frac{\epsi^3}{2\cdot 6 \pi} \, | \ddot d (t)|^2 \;+ \;{o}(\epsi^3)\,.
\]

To come back to our main message once more, note that computing naively the energy in the piece of the solution that constitutes the error of the adiabatic approximation $\tilde \psi (t) := \psi(t) -\psi_{\rm ad}(t)$ would have led to
 \[
 \langle \tilde \psi (t), H_{\rm f} (t) \tilde \psi (t) \rangle_\fock \;=\;\frac{\epsi^2}{4} \sum_{i,j=1}^N \int_{\R^3}\D k\,
 \frac{e_ie_j|\hat\varphi(k)|^2}{|k|^2} \E^{\I k\cdot (x_j(t)-x_i(t))} \kappa \cdot \dot x_j(t) \kappa \cdot \dot x_i(t)
+  \Or(\epsi^3)\,,
 \]
 where $\kappa := k/|k|$. As our analysis showed, this velocity-dependent term is the energy in the deformation of the dressing and being of order $\epsi^2$, it is much bigger than the energy emitted by radiation of free bosons, which is of order $\epsi^3$. In order to obtain the correct picture, the use of superadiabatic approximations seems inevitable. 
 
 Finally let us mention that a closely related problem was considered in~\cite{tenteu}, see also \cite{ten}.
 There semiclassical limit  of   non-relativistic Schr\"odinger particles coupled to a massless scalar field was considered. In a sense the system considered in the present paper can be considered a toy model version of the one in~\cite{tenteu}. However, as a consequence of the simplicity of the model the results obtained here are much stronger and at the same time the proofs are more transparent. For 
 \\
 
 \noindent{\bf Acknowledgements.} We thank Rainer Nagel for pointing out to us reference~\cite{Kato}.

 \section{Main results} 
 
 We always assume the following properties for the parameters of the model:
 
\begin{asum}
  The charge distribution $\varphi \in S(\R^3)$ is a spherically symmetric Schwartz-function  with $\hat\varphi(0)=(2\pi)^{-3/2}$ and 
  $x_j\in C^4(\R,\R^3)$ for $j=1,\ldots, N$.
\end{asum} 
\noindent    
As explained in the introduction, the central object that we construct is the transformation to the superadiabatic representation, a
    unitary dressing transformation 
\[
V^\epsi (t): \fock \to \fock\,.
\] 
In the superadiabatic representation the 
    Schr\"odinger equation  reads
\[
V^\epsi (t)\left(\I \epsi \tfrac{\D}{\D t}-H(x(t))\right)V^{\epsi  } (t)^* \,V^\epsi (t) \, \psi(t)=: \Big(\I\epsi \frac{\D}{\D t}-H^\epsi_{\mathrm{dress}}(t) \Big)\phi(t)\,.
\]
In this representation the Fock vacuum $\Omega_0$ corresponds to the dressed ground state and the $m$-particle sectors of Fock space correspond to states with $m$ free bosons. 
While in the introduction we formulated the statements in the original  representation, it is more convenient to formulate and prove everything in the new representation after performing the dressing
transformation. The translation back to the original representation is straightforward using the definitions of $b^{\epsi \#}_t$ of the introduction. Note, however, that if the total charge is different from zero, the transformation 
$V^\epsi(t)$ has no asymptotic expansion in powers of $\epsi$, not even the limit $\lim_{\epsi\to 0} V^\epsi(t)$ exists.

Let $Q_{\leq M}$ denote the projection on the sectors of Fock space with at most $M$ photons.

\begin{thm}[\bf Adiabatic approximation]\label{thm:adapp}

The dressed Hamiltonian $H_\mathrm{dress}^\epsi$ generates a unitary propagator $U_\mathrm{dress}(t):= U_\mathrm{dress}(t, t_0=0)$. For any $T\in \R$ there are constants $C<\infty$ and $\epsi_0>0$ such that for all   $M\in \N$, $0<\epsi\leq\epsi_0$ and $|t|\leq T$
\[
\left\| \left( U_\mathrm{dress}(t)-  \ee^{-\frac{\im}{\epsi}(H_{\mathrm{f}} t +\int_0^t   E (s) \,\D s)}\right)Q_{\leq M} \right\| \leq    C\,\epsi \, {\ln  (\epsi ^{-1}) \,\sqrt{M+1}} 
\]
   in the norm of $\mathcal{L}(\fock)$ and in the norm of $\mathcal{L}(D(H_\mathrm{f}))$.
Here 
\begin{equation}\label{EDef}
E( t ):=-\textstyle\frac{1}{2}\int_{\R^3} \D k\, \frac{ |v_0 (x(t),k)|^2 }{|k|}
\end{equation}
with
\begin{equation}\label{equ:v(x(t),k)}
v_0(x(t),k):= \sum_{j=1}^N  \frac {e_j \hat \varphi(k)}{|k|^{{1 /2}}}\ee^{ \im k \cdot x_j(t)}\,.
\end{equation}
\end{thm}
Thus at leading order the time evolution is just the free evolution of the bosons and a time-dependent dynamical phase. In particular, the $m$-boson sectors $Q_m\fock$  are adiabatic invariants of the dynamics. An initial state in $Q_m\fock$ of the form
\[
\phi(0) =\textstyle \prod_{l=1}^m a^\dagger (f_l) \,\Omega_0\,,\quad f_l\in L^2(\R^3)\,,
\]
  evolves into
\[
\phi(t) := U_\mathrm{dress}(t) \,\phi(0) =  \ee^{-\frac{\im}{\epsi} \int_0^t \D s\, E(s) }\textstyle \prod_{l=1}^m a^\dagger (f_l(t)) \,\Omega_0\,.
\]
Applying $V^\epsi (t)$ and using 
\[
\| V^\epsi (t) - V(x(t)) \| =\Or(\epsi)
\]
in the case of a neutral system  yields (\ref{adiapprox}).

\begin{thm}[\bf Superadiabatic approximation]\label{thm:foae}

For any $T\in \R$ there are constants $C<\infty$ and $\epsi_0>0$ such that for all   $M\in \N$, $0<\epsi\leq\epsi_0$ and $|t|\leq T$
\begin{align*}\lefteqn{
\left\| \left( U_\mathrm{dress}(t)  - \ee^{-\frac{\im}{\epsi}(H_{\mathrm{f}} t +\int_0^t   E^\epsi(s) \,\D s)}\left( 1+ \im \epsi    
  \int_0^t \D s\,\ee^{\im H_{\mathrm f} \frac{s}{\epsi}}\,h_{\mathrm{rad}}(s)\, \ee^{-\im H_{\mathrm f} \frac{s}{\epsi}} \right)\right)Q_{\leq M} \right\|\leq } \\[2mm] &&&&&&&&&&&&&&&&&&\leq
    C\,\epsi^2\,  \ln  (\epsi ^{-1} )\,\sqrt{M+1}   
  \end{align*}
%
 in $\mathcal{L}(\fock)$ and $\mathcal{L}(D(H_\mathrm{f}))$ with  
\[
E^\epsi(t) = E(t) +  \frac{\epsi^2}{4} \sum_{j,i=1}^N    \int_{\R^3}\D k \,  \frac{e_je_i\betr{\hat {\varphi}(k)}^2}{\betr{k}^2}  \ee^{-\im k\cdot (x_i(t)-x_j(t))}\,(\kappa\cdot \dot{x}_j(t))  \,(\kappa \cdot \dot{x}_i (t))
\]
and
\[
h_\mathrm{rad}(t):=\Phi\Big( \id_{[\epsi^8,\infty)}(|k|) \sum_{j=1}^N   \frac{ e_j \hat \varphi (k) }{|k|^{\frac{3}{2}}}\ee^{-\im k \cdot x_j(t)}  \langle  \kappa,\ddot x_j(t) \rangle \Big)\,.
\]
\end{thm}

Note the infrared regularization by the characteristic function $ \id_{[\epsi^8,\infty)}(|k|)$  cutting off boson momenta smaller than $\epsi^8$ in this definition, which can be omitted for  neutral systems. 

With the help of this theorem we can define a decomposition of the wave function into a superadiabatic part $\phi_{\mathrm{su}}$ and the non-adiabatic part $\phi_\mathrm{na}$ with different photon number.
For a wavefunction $\phi_0 \in Q_m \fock$ we define
\[
\phi_{\mathrm{su}}(t) := Q_m \phi(t)= \ee^{-\frac{\im}{\epsi}(H_{\mathrm{f}} t +\int_0^t   E^\epsi(s) \,\D s)} \phi_0+\mathcal{O}\big(  \epsi^2  \ln  (\epsi ^{-1} )\big)
\]
and 
\[
\phi_{\mathrm{na}}(t) := Q_m^\perp \phi(t)= \im \epsi  \ee^{-\frac{\im}{\epsi}(H_{\mathrm{f}} t +\int_0^t   E^\epsi(s) \,\D s)}
 \int_0^t \D s\,\ee^{\im H_{\mathrm f} \frac{s}{\epsi}}\,h_{\mathrm{rad}}(s)\,\ee^{-\im H_{\mathrm f} \frac{s}{\epsi}} \phi_0  + \mathcal{O}\big( \epsi^2  \ln  (\epsi ^{-1} )\big)\,.
\]

Our final result concerns the amount of energy lost by the system due to radiation. For the stationary problem
with Hamiltonian $H(x)$ independent of time, the natural definition for the energy carried by free photons in the state $\psi\in D( H_{\rm f})$ is
\[
E_{\rm rad,stat}(\psi)  \;:=\; \langle \psi , H(x) \psi \rangle - E(x)\,,
\]
where $E(x)$ is the ground state energy. For  time-dependent Hamiltonians
the definition of energy is somewhat subtle, more so, because we look at the energy of a subsystem, the free photons.
However, in the superadiabatic 
representation there is again a natural definition for the energy given by the free photons, namely
\begin{align}\label{equ:rad}
E_{\rm rad}(\psi(t))  \;:=\; \langle  V^\epsi (t)\,\psi , H_{\rm f}  \, V^\epsi (t)\,\psi  \rangle  \,.
\end{align}
We will show that at any  time $t$     when $\dot x(t) = \ddot x(t)=0$, it holds that 
\begin{equation}\label{energy}
E_{\rm rad}(\psi(t)) = E_{\rm rad, stat}(\psi(t))  + \Or(\epsi^4)\,.
\end{equation}
Thus in a situation where the sources move only during a finite time interval, before and after the change the definition of energy agrees with the static one. Note that the error of order $\epsi^4$ comes from the infrared regularization in the transformation $V^\epsi(t)$ and is identically zero for neutral systems. The equality (\ref{energy}) is also determining the exponent $8$ in the infrared regularization. A smaller exponent there would lead to a larger error in  (\ref{energy}).

If we  assume that the initial state is the dressed vacuum, then the  energy  emitted as a result  of the motion of the sources is just the energy of the free photons and thus
$
{E}_{\mathrm{rad}}(t) =  \langle \phi_{\rm na}(t ), H_{\rm f} \phi_{\rm na}(t ) \rangle  
$,
which is computed explicitly in the following theorem.

\begin{thm}[\bf Radiation]\label{thm:rad}
Let $\phi(0)=\Omega $, then uniformly on bounded intervals in time it holds that
\begin{align} \label{equ:rad1}
E_{\mathrm{rad}}(t ) &\;=\; \frac{\epsi^2}{2} \sum_{i,j=1}^N {e_j e_i} \int  \D k\,   \frac{|\hat \varphi(k)|^2}{|k|^2}
\int_0^{t } \int_0^{t } \D s \, \D s'\,  \ee^{\im |k|\frac{s-s'}{\epsi}}\ee^{-\im k \cdot
(x_j(s)-x_i(s'))}\\
&\quad \qquad\qquad\times  \kappa\cdot \ddot x_j(s) \,
 \kappa \cdot \ddot x_i(s')   + \mathcal O\left(  \epsi^4  \big(\ln  (\epsi ^{-1} )\big)^2 \right)  \nonumber\\
 &\;=\;\frac{\epsi^3}{2\cdot 6 \pi}\int_0^t \D s\, | \ddot d (s)|^2 \; +\; o(  \epsi^3 )\,,\label{equ:rad2}
\end{align}
  where 
 \[
 \ddot d(t) := \textstyle\sum_{j=1}^N e_j \ddot x_j(t) 
 \]
 is the second derivative of the dipole moment.
\end{thm}
%
%
%
%
%

\section{Mathematical details}\label{sec:preq}

\subsection{Operators on Fock space}
\label{sec:fock}

In this section we recall the   definitions of the operators on Fock space and some of their important properties.
 Proofs of all claims can be found in~\cite{resi2}, Section X.7. 

We call $\mathcal{F}_{\mathrm{fin}} $ the subspace of the Fock space for which $ \psi_m =0$ for all but finitely many~$m$.
The second quantization of a self-adjoint multiplication operator $\omega$ with domain $D(w)\in L^2(\R^3)$ is defined for $ \psi \in \mathcal{F}_{\rmm {fin}} $ with components $ \psi_{m} \in \otimes_{k=1}^m D(w) $   as
\[
 (\D \Gamma(\omega)\psi)_{m}(k_1,\dots,k_m)=\textstyle\sum_{j=1}^m\omega(k_j) \psi_{m}(k_1,\dots,k_m)
\]
and is essentially self-adjoint. In particular, the free field Hamiltonian $H_\mathrm{f}:=\D \Gamma(|k|)$ is self-adjoint on its maximal domain.
The annihilation operator and the creation operator on $ \psi \in \mathcal{F}_{\rmm {fin}} $ are defined by
\begin{align*}
 (a(f)\psi)_{m}(k_1,\dots, k_m)&:=\sqrt{m+1}\textstyle\int_{\mathbb R^3} \D k\, \bar f(k) \psi_{m+1}(k,k_1,\dots,k_m)\,,\\
 (a^\dagger(f)\psi)_m(k_1,\dots,k_m )&:=\frac{1}{\sqrt{m}}\textstyle \sum_{j=1}^m f(k_j) \psi_{m-1}(k_1,\dots,\tilde k_j,\dots ,k_m)\,.
\end{align*}
Here $ \tilde k_j$ means that $k_j $ is omitted. They fulfill the canonical commutation relations
\[
[a(f),a^\dagger(g)]=\langle f,g \rangle_{L^2(\R^3)}\,, \  \
[a(f),a(g)]=0\,, \  \ [a^\dagger(f),a^\dagger(g)]=0\,.
\] 
The creation and annihilation operator define the Segal field operator
\[
\Phi(f):= \tfrac{1}{\sqrt{2}}\big(a(f)+a^\dagger(f)\big)\,.
\] 
It is essentially self-adjoint on  $ \mathcal{F}_{\mathrm{fin}}$. 
The canonical commutation relations imply 
\begin{align}\label{ccr1}
 [\Phi(f),\Phi(g)]&=\im \, \mathrm{Im} \langle f,g \rangle_{L^2(\R^3)} \quad\mbox{and}\quad [\D \Gamma(\omega),\im \Phi(f)]=\Phi(\im \omega f)\,.
\end{align}
 
%
%
\subsection{The Hamiltonian}
As described in the introduction we consider the Hamiltonian
\[
H(t) := H(x(t)) = H_{\rm f} + \Phi( v_0(x(t),k)) =: H_{\rm f} + \Phi_0(t)
\]
with $v_0(x(t),k)$ defined in (\ref{equ:v(x(t),k)}).
As to be stated more precisely in Lemma~\ref{lem:spc} and after, this Hamiltonian has  a ground state only for a neutral system, i.e.\ if $\sum_{j=1}^N e_j = 0$,  otherwise   the bottom of the spectrum is not an eigenvalue.
To deal with this fact we introduce an infrared cutoff $0\leq \sigma\leq 1$ in the interaction and
put $\hat \varphi_\sigma(k) =   \id_{[\sigma,\infty)}(|k|) \hat \varphi (k) $  and
$v_\sigma(x ,k) =   \id_{[\sigma,\infty)}(|k|) v_0(x ,k)$.
The resulting truncated Hamiltonian $H^\sigma(t)= H_{\rm f} + \Phi_\sigma(t) $ is a good
approximation to the original Hamiltonian as will be proven in Proposition~\ref{prop:approxs}. We emphasize that $H^\sigma(t)$ is used only as a tool in the proofs and the final results of Section~2 hold for the  Hamiltonian $H(t)=H^{\sigma =0}(t)$ without infrared cutoff.  

\begin{lem}
\leavevmode
 The operators   $H^{\sigma}(t)$ are self-adjoint on $D(H_{\mathrm f})$ for all $t$ and $\sigma\in[0,1]$.
   The graph norms of   $H^{\sigma}(t)$ are all equivalent  to the one defined by $ H_{\mathrm f}$ uniformly in  $\sigma$.
\end{lem}

\begin{proof}
 There exists a standard estimate for $\Phi$ (cf.\ e.g.~\cite{betz} Proposition 1.3.8). Let $\Psi \in $ $D(H_\rmm{f})$ and $b > 0$, then
\[
\norm{\Phi(f) \psi}^2_{\mathcal{F}} \leq b \norm{H_{\mathrm f} \psi}_{\mathcal{F}}^2+ \left(\tfrac{1}{b}\text{$\bigl\|f/\sqrt{|k|}\bigr\|^4_{L^2(\R^3)}$} +2 \norm{f}^2_{{L^2}(\R^3)} \right)\norm{\psi}_{\mathcal{F}}^2\,.
\]
Thus $\Phi(f)$ is infinitesimally bounded with respect to $ H_{\mathrm f}$
if 
\begin{align} \label{ineq11}
\norm{f}_{{L^2}(\R^3)}+ \| f/\sqrt{|k|}\|_{{L^2}(\R^3)} < \infty\,.
\end{align}
Equation \eqref{ineq11} holds uniformly for $f =  v_\sigma(x(t),k)$ and $\sigma \in [0,1]$.  Hence $H^\sigma(t)$ is self-adjoint on $ D(H_{\mathrm f})$ by the Kato-Rellich Theorem (cf.~\cite{resi2}, Theorem X.12)  for all $\sigma \in [0,1]$ and the graph norms are all equivalent.
\end{proof}

\begin{prop} \label{prop:U}
\leavevmode Let $t \mapsto x(t)\in C^n(\R,\R^{3N}) $, $n\in\N$, and equip $D(H_\rmm{f})$ with the graph norm.
Let $\sigma\in[0,1]$, then
\begin{enumerate}
 \item        $ H^\sigma \in C^n_b(\R,\mathcal{L}(D(H_\rmm{f}),\mathcal{F}))$
\item  $H^\sigma$   generates a  strongly continuous  unitary   evolution family denoted by   $U_{H^{\sigma}}(t,t_0)$. Moreover,
  $U_{H^{\sigma}}(t,t_0)$ is a   bounded operator  on  $D(H_{\rmm{f}})$ with 
\begin{equation}\label{UDNorm}
  \sup_{s\in[t_0,t]} \norm{U_{H^{\sigma}}(s,t_0)}_{\mathcal{L}(D(H_{\rm f}))} <\infty
\end{equation}
for any $t\in\R$ and uniformly in $\sigma$.
\end{enumerate}
\end{prop}

\begin{proof}
According to \eqref{ineq11} we need to show that 
 \[
 \norm {\tfrac{\D^n}{\D t^n} v_0(x(t),k)} + \norm {\tfrac{\D^n}{\D t^n} v_0(x(t),k)/\sqrt{|k|}} < \infty \,,
 \] 
 which obviously holds under our assumptions.
That $H$ and $H^\sigma$ generate unitary evolution families follows from the first statement and  
 the general result about contraction semigroups (e.g.~\cite{resi2}, Theorem X.70).
 The norm bound  $  U_{H^{\sigma}}(t,t_0)$  in $\mathcal{L}(D(H_{\rm f}))$ is less known and was proved in Theorem~5.1.(e) of~\cite{Kato}.
\end{proof}

\begin{lem}\label{lem:spc}
Let $K =H_{\mathrm f}+\Phi(z( k))$ with  $z(k)\in L^2(\R^3)$ such that
\[
 \alpha := -\tfrac{1}{2}\textstyle\int_{\R^3} \D k\, \frac{\betr{z( k)}^2}{|k|} < \infty  \,.
\]
Then the spectrum of $K $ is given by $[\alpha ,\infty)$ and the infimum of the spectrum  $\alpha $  is an eigenvalue if and only if
\[
\textstyle \int_{\R^3} \D k\, \frac{\betr{z( z)}^2}{|k|^2} < \infty \,. 
\]
In this case  the unitary operator 
$ V  = \ee^{-\im \Phi(\im z( k)/|k| )}$ is well-defined   and 
$ K =V^* H_{\mathrm f}V +\alpha$.
\end{lem}

\begin{proof}
The first part is Proposition 3.10 of~\cite{dere} and the second part is Proposition 3.13 of the same paper.
\end{proof}

\begin{cor} The spectrum of $H(t)$ is
 $\sigma (H(t))=\sigma_{\mathrm{ac}}(H(t))=[E(t),\infty)$ with $E(t)$ as in (\ref{EDef}).
\end{cor}

\begin{cor}\label{cor:spc}
 For $\sigma>0$, the infrared regularized   Hamiltonian  $H^{\sigma}(t)$ can be written as
\[
 H^{\sigma}(t)= V^*_\sigma(t)H_{\mathrm f} V_\sigma(t)+E_\sigma(t)\,,
\]
where 
\[
 V_\sigma(t):=\ee^{-\im \Phi(\im v_\sigma(x(t),k)/|k| )}\,.
\]
Its only eigenvalue with eigenvector $\Omega_\sigma(t):= V_\sigma^*(t) \Omega_0   $ is
\begin{equation}\label{Esigma}
E_\sigma(t) = E(t)+\mathcal{O}(\sigma)\,.
\end{equation}
\end{cor}

\begin{proof}
 The last statement is an immediate consequence of Lemma~\ref{lem:spc} and the fact that
\[
 \int_{\R^3}  \D k\,   \frac{ \betr{v_\sigma(x(t),k)}^2}{|k|} =  \int_{\R^3}  \D k\,   \frac{\betr{v_0(x(t),k)}^2}{|k|}+\mathcal{O}(\sigma) \,.   \qedhere
\]
\end{proof}

\subsection{Superadiabatic perturbation and the dressing operator}
In this section we introduce the perturbed projections and the transformation which ``diagonalizes`` the infrared regularized Hamiltonian up to $\mathcal{O}(\epsi^2)$.
As explained in the introduction the idea is to modify the adiabatic projections
$P_0^m(t) := V_\sigma^*(t) Q_m V_\sigma(t)$ on the $m$-free-photons subspaces in such a way, that the modified superadiabatic projections 
$P_1^{m,\epsi}(t)$ are higher order adiabatic invariants, i.e.\ that 
\[
 \left[\im \epsi \tfrac{\D }{\D t}-H(t),P_1^{m,\epsi}(t) \right]=\mathcal{O}(\epsi^{2}) 
\]%
holds. This construction is by now standard, see~\cite{Nenciu}, and yields
\[
  P_1^{m,\epsi}(t) = V_\sigma^*(t) Q_m V_\sigma(t)+\im \epsi V_\sigma^*(t)[Q_m, \Phi_2(t)]  V_\sigma(t) + \Or(\epsi^2)\,,
\]
where we introduce  
  the shorthands 
\begin{align*}
\Phi_1(t) & := \tfrac{\D}{\D t} \Phi\Big(\frac{\im v_\sigma(x(t),k)}{|k|}\Big) = \Phi \Big(- |k| \sum_{j=1}^N  \frac {e_j \hat \varphi_\sigma(k)}{|k|^{{3 /2}}}\;\ee^{ \im k \cdot x_j(t)} \langle \kappa,\dot x_j(t)\rangle\Big)=:\Phi(z_1(t))  \\
\Phi_2(t)&:= -\im [\D \Gamma(|k|^{-1}), \Phi_1] = \Phi\Big(\im  \sum_{j=1}^N  \frac {e_j \hat \varphi_\sigma(k)}{|k|^{{3 /2}}}\;\ee^{ \im k \cdot x_j(t)} \langle \kappa,\dot x_j(t)\rangle\Big)=:\Phi(z_2(t))\,.
\end{align*}
Here $\kappa$ is the unit vector in the direction of $k$. The canonical commutation relations link $\Phi_1$ and $\Phi_2$ in the following way
\[
 \im [\D \Gamma(|k|),\Phi_2]=\Phi_1\,.
\]
Now the corresponding  dressing operator $ V^\epsi_\sigma(t)$ that maps $ P_1^{m,\epsi}(t)$ to $Q_m$ then needs to have the expansion
\[
 V^\epsi_\sigma(t) = (\id+ \im \epsi \Phi_2(t) + \Or(\epsi^2)) V_\sigma(t) \quad \mathrm{and} \quad V^\epsi_\sigma(t)^* = V^*_\sigma(t)(\id- \im \epsi  \Phi_2(t)+ \Or(\epsi^2))\,. 
\]
This suggests to define 
\[
V^\epsi_\sigma(t) := \E^{\I\epsi\Phi_2(t)} V_\sigma(t)\,.
\]

\begin{lem}\label{lem:V}
For $\sigma\in(0,1]$  and $\epsi\in[0,1]$ the operator 
 $V^\epsi_\sigma(t)$ is unitary, belongs to $\mathcal{L}(D(H_\rmm{f}))$ and satisfies
 \[
 \norm{V^\epsi_\sigma(t)}_{\mathcal{L}(D(H_\rmm{f}))} +  \norm{V^\epsi_\sigma(t)^*}_{\mathcal{L}(D(H_\rmm{f}))}\leq C <\infty
\]
uniformly in  $\sigma$ and $\epsi$.\\[2mm]
The map $\R\to \mathcal{L}(D(H_{\rm f}),\mathcal{F})$,  $t\mapsto V^\epsi_\sigma(t)$ 
is differentiable and
\begin{equation}\label{abl}
\tfrac{\D}{\D t} V^\epsi_\sigma  \;=\; \left(-\I\Phi_1 +\I\epsi \dot \Phi_2 + \epsi[\Phi_2,\Phi_1]-\tfrac{\epsi^2}{2} [\Phi_2,\dot\Phi_2] \right)\, V^\epsi_\sigma \,.
\end{equation}
\end{lem}

\begin{proof}
With 
\[
[ H_{\rm f}, V_\sigma(t)]\;=\;V_\sigma(t) \left( V_\sigma^*(t) H_{\rm f} V_\sigma(t) - H_{\rm f}\right) \;=\; V_\sigma(t) (\Phi_0(t)-E_\sigma(t))
\]
and 
\begin{eqnarray*}
[ H_{\rm f}, \E^{\I\epsi\Phi_2}]&=&\E^{\I\epsi\Phi_2} \left( \E^{-\I\epsi\Phi_2} H_{\rm f} \E^{\I\epsi\Phi_2}- H_{\rm f}\right) \\&=&  \E^{\I\epsi\Phi_2} \Big( \sum_{j=0}^\infty \epsi^j\,\frac{[-\I\Phi_2, [-\I\Phi_2 , [-\I\Phi_2,\cdots  H_{\rm f}]]]]}{j!}-H_{\rm f} \Big)  \\
&= & \E^{\I\epsi\Phi_2}\left( \epsi \Phi_1 - \I \tfrac{\epsi^2}{2} {[\Phi_2, \Phi_1]} \right)
\end{eqnarray*}
the first statement follows, since
  $\Phi_0$, $\Phi_1$ and $[\Phi_2, \Phi_1]= \I\,{\rm Im}\langle z_2 ,z_1 \rangle$ are bounded independently of $\sigma$ as operators from $D(H_{\rm f})$ to $\mathcal{F}$. 

Under our hypotheses on $\hat\varphi$  the operators
 $ \Phi(\im v_\sigma(x(t),k)/|k|)$ and $\dot \Phi = \Phi_1$ commute as a consequence  of (\ref{ccr1}) 
  and  $ \langle v_\sigma(x(t),k) / |k|, v_\sigma(x(t),k)   \kappa\cdot \dot x_j(t) \rangle =0$.
Hence 
\[
\tfrac{\D}{\D t} V_\sigma(t) \;= \; \tfrac{\D}{\D t} \E^{-\I\Phi(\im v_\sigma(x(t),k)/|k|)}\;=\; -\I\dot \Phi(t) V_\sigma(t) = -\I \Phi_1(t) V_\sigma(t)\,.
\]
For the other exponential we find
\begin{eqnarray*}
\tfrac{\D}{\D t}  \E^{\I\epsi\Phi_2} &=& \left[ \partial_t , \E^{\I\epsi\Phi_2} \right] \;=\;
\left( \partial_t - \E^{\I\epsi\Phi_2}\partial_t\E^{-\I\epsi\Phi_2}\right) \E^{\I\epsi\Phi_2}
\;=\; \left( \I\epsi\dot\Phi_2 -\tfrac{\epsi^2}{2} [\Phi_2,\dot\Phi_2] \right)\E^{\I\epsi\Phi_2}\,.
 \end{eqnarray*}
Finally
\[
\E^{\I\epsi\Phi_2} \Phi_1\E^{-\I\epsi\Phi_2} \;=\;\Phi_1 +\I\epsi [\Phi_2,\Phi_1]
\]
implies (\ref{abl}).
\end{proof}

\subsection{The dressed Hamiltonian}\label{sec:dressn}
With the help of the dressing transformation $V^\epsi_\sigma(t)$, we   define the dressed Hamiltonians  
\[
\begin{array}{rcl}
 \im \epsi \partial_t -H_{\mathrm{dress}}(t)&:=&V^\epsi_\sigma(t) \big(\im \epsi \partial_t -H (t) \big)V^\epsi_\sigma(t)^*\,,\\[3mm]
 \im \epsi \partial_t -H^\sigma_{\mathrm{dress}}(t)&:=&V^\epsi_\sigma( t) \big(\im \epsi \partial_t -H^{\sigma}(t) \big)V^\epsi_\sigma(t)^*\,.
 \end{array}
\]
As a consequence of Lemma~\ref{lem:V} the dressed Hamiltonians are self-adjoint on $D(H_\rmm{f})$ since $V^\epsi_\sigma(t)$ is a bijection on $D(H_\rmm{f})$. 
The corresponding evolution families are related through
\[
U_{ {\rm dress}}(t) = V^\epsi_\sigma(t) \,U_{H}(t) V^\epsi_\sigma(0)^*\qquad \mbox{and}\qquad U^\sigma_{ {\rm dress}}(t) = V^\epsi_\sigma(t) \,U_{H^\sigma}(t)V^\epsi_\sigma(0)^*\,.
\]

\begin{lem}\label{lem:hdress}
 The  dressed Hamiltonian has the form
\[
H^\sigma_{\mathrm{dress}}(t)\;=\;
H_{\rm f} + E^\epsi_\sigma(t) - \epsi^2 \dot\Phi_2(t)
\]
with
\[
E^\epsi_\sigma(t) :=  E_\sigma(t) -\tfrac{ \epsi^2}{2} {\rm Im}\langle z_2 (t),z_1(t)\rangle
 +\tfrac{ \epsi^3}{2} {\rm Im}\langle z_2(t),\dot z_2(t)\rangle\,.  
\]
\end{lem}

\begin{proof}
 For $H^{\sigma }(t)$ we find
 \begin{eqnarray*}
V^\epsi_\sigma \,H^{\sigma }\,V^{\epsi*}_\sigma &=&  \E^{\I\epsi\Phi_2 } V_\sigma 
  H^{\sigma }  V^*_\sigma  \E^{-\I\epsi\Phi_2 }  
  \;\;=\;\;\E^{\I\epsi\Phi_2  } (H_{\rm f} +E_\sigma) \E^{-\I\epsi\Phi_2 } \\
  & =&  \sum_{j=0}^\infty \epsi^j\,\frac{[\I\Phi_2, [\I\Phi_2 , [\I\Phi_2,\cdots  H_{\rm f}]]]]}{j!}+ E_\sigma   \\
 &=& H_{\rm f} + E_\sigma - \epsi \Phi_1 - \I \tfrac{\epsi^2}{2} {[\Phi_2, \Phi_1]} \;\;=\;\;
   H_{\rm f} + E_\sigma - \epsi \Phi_1 + \tfrac{\epsi^2}{2} {\rm Im} \langle z_2,z_1\rangle\,,
 \end{eqnarray*}
and with (\ref{abl})
\[
V^\epsi_\sigma \,\I\epsi\partial_t\,V^{\epsi*}_\sigma \;=\; \I\epsi\partial_t  - \I\epsi (\tfrac{\D}{\D t} V^\epsi) V^{\epsi*}\;=\; \I\epsi\partial_t -\epsi \Phi_1 + \epsi^2 \dot \Phi_2 - \I\epsi^2 [\Phi_2,\Phi_1]+\tfrac{\I \epsi^3}{2} [\Phi_2,\dot\Phi_2]\,.
\]
Hence 
\[
H_{\rm dress}^\sigma(t) \;\;=\;\; H_{\rm f} + E_\sigma(t) -{\epsi^2}\dot \Phi_2(t) - \tfrac{\epsi^2}{2} {\rm Im}\langle z_2(t),z_1(t)\rangle+ \tfrac{\epsi^3}{2} {\rm Im}\langle z_2(t),\dot z_2(t)\rangle \,.\qedhere
\]
\end{proof}

\begin{prop}\label{prop:approxs} For any $\sigma,\epsi\in(0,1]$ and $M\in\N$ we have
\[
 \norm{\Big(U_\mathrm{dress}(t)
-U^{\sigma }_\mathrm{dress}(t)\Big) Q_{\leq M} } = \Or\left(\sqrt{M+1}\,\frac{\sigma}{\epsi}\right)
 \]
uniformly on bounded intervals in time  in $\mathcal{L}(\fock)$ and $\mathcal{L}(D(H_\mathrm{f}))$.
\end{prop}

 \begin{proof} By definition we have 
\[
 \Big(U_\mathrm{dress}(t)
-U^{\sigma }_\mathrm{dress}(t)\Big) Q_{\leq M}     = V^\epsi_\sigma(t) \,(U_{H}(t) -U_{H^\sigma}(t) )\,V^\epsi_\sigma(0)^*\,Q_{\leq M} \,.
\]
Let 
\[
H(s)-H^{\sigma}(s) =  \Phi \big(\id_{(0,\sigma)}(|k|)\,  v_0(x(s),k)\big) =: \Phi(f_{\leq\sigma}(s))\,.
\]
Then bosons with momenta smaller than $\sigma$ evolve freely under the regularized evolution, i.e.\ 
\[
U_{H^\sigma}^*(t)\,\Phi(\E^{ \I |k| \frac{t}{\epsi} }f_{\leq\sigma}(s))\,U_{H^\sigma}(t) = \Phi(f_{\leq\sigma}(s))
\]
holds for all $t\in\R$. To see this note that the equality holds  for $t=0$ and the derivative of the left hand side vanishes since 
\[
[ H^\sigma , \Phi(\E^{ \I |k| \frac{t}{\epsi}} f_{\leq\sigma}(s)) ] = [ H_{\rm f} , \Phi(\E^{ \I |k| \frac{t}{\epsi} }f_{\leq\sigma}(s)) ] = -\I \epsi \tfrac{\D}{\D t} \Phi(\E^{ \I |k| \frac{t}{\epsi}} f_{\leq\sigma}(s))\,.
\]
Moreover, $\Phi(f_{\leq\sigma}(s))$ commutes with $V^\epsi_\sigma(0)^*$. 
 By the Duhamel formula we thus find  
 \begin{eqnarray*} \lefteqn{\hspace{-20mm}
 \norm{(U_{H}(t )-U_{H^{\sigma}}(t ))\,V^\epsi_\sigma(0)^*\,Q_{\leq M} } 
 \leq  \frac{1}{\epsi}\int_{ 0}^t \D s\, \norm{(H(s)-H^{\sigma}(s))U_{H^{\sigma}}(s ) V^\epsi_\sigma(0)^*Q_{\leq M} }} \\
 &= &\frac{1}{\epsi} \int_{ 0}^t \D s\,\norm{\Phi(f_{\leq\sigma}(s))  U_{H^{\sigma}}(s ) V^\epsi_\sigma(0)^*Q_{\leq M} }  \\
  &= &\frac{1}{\epsi} \int_{ 0}^t \D s\,\norm{ U_{H^{\sigma}}(s ) \Phi(-\E^{ \I |k| \frac{s}{\epsi}} f_{\leq\sigma}(s))  V^\epsi_\sigma(0)^*Q_{\leq M} } \\
  &= &\frac{1}{\epsi} \int_{ 0}^t \D s\,\norm{  \Phi(-\E^{ \I |k| \frac{s}{\epsi}} f_{\leq\sigma}(s))  Q_{\leq M} }\,.
 \end{eqnarray*}
The following lemma, which will be applied several times in the following, together with the fact that
\[
\|f_{\leq\sigma}(s)\|_{L^2} + \||k| f_{\leq\sigma}(s)\|_{L^2}   = \Or(\sigma) 
\]
uniformly for $s\in[0,t]$
allows us to conclude.
\end{proof}

\begin{lem}\label{philem}
There is a constant $C<\infty$  such that for any  $f\in L^2(\R^3)$ and $M\in \N$ 
\[
\| \Phi(f) \,Q_{\leq M} \|_{\mathcal{L}(\fock)} \leq 2^\frac{1}{2} \sqrt{M+1} \norm{f }_{L^2(\R^3) }
 \]
 and
 \[
  \| \Phi(f) \,Q_{\leq M} \|_{\mathcal{L}(D(H_{\mathrm f}))}   \leq C \sqrt{M+1} \left(\norm{f }_{L^2(\R^3) } +\norm{|k|f }_{L^2(\R^3) }    \right)
\]

\end{lem} 

\begin{proof}
 According to~\cite{resi2}, Theorem X.41, we have for $f\in L^2(\R^3)$ that
\[
\norm{\Phi(f )  Q_{\leq M} \psi}_{\mathscr{\fock}} \leq 2^\frac{1}{2} \sqrt{M+1} \norm{f }_{L^2(\R^3) } \norm {\psi}_{\mathscr{\fock}} \,.
\] 
The second claim follows from the observation that
\[
H_{\rm f} \Phi(f )  Q_{\leq M}  =   \Phi(f )  Q_{\leq M} H_{\rm f} \;-\; \I \Phi(\I |k| f) Q_{\leq M}
\]
together with the first estimate.
\end{proof}

 \subsection{Effective dynamics}
  
In this section we  first show that the statements of Theorems~\ref{thm:adapp} and~\ref{thm:foae} hold for 
the infrared regularized evolution $U^{ \sigma}_\mathrm{dress}(t)$ with an error depending only logarithmically on~$\sigma$.
Then we use Proposition~\ref{prop:approxs}  and an appropriate choice for $\sigma=\sigma(\epsi)$ to show   the statements also for the full evolution $U_\mathrm{dress}(t)$. Clearly the statement of Theorem~\ref{thm:adapp} is a consequence of Theorem~\ref{thm:foae}, but since the proof necessarily  proceeds in the same two steps, we separated the statements. 
 
 \begin{lem} For any $0<\sigma\leq\frac{1}{3}$ and $M\in\N$  it holds that
  $  \dot \Phi_2 Q_{\leq M} \in \mathcal{L}(D(H_\mathrm{f}))\cap \mathcal{L}(\fock) $
  with 
  \[
   \|\dot \Phi_2 Q_{\leq M}\| = \Or(\sqrt{M+1}\ln(\sigma^{-1}))
  \] 
 uniformly on bounded time intervals in $\mathcal{L}(\fock)$ and $\mathcal{L}(D(H_\mathrm{f}))$.
 \end{lem} 
  \begin{proof}
 Since 
\begin{equation}\label{v2dot}
\| \dot z_2 \|^2_{L^2(\R^3)} = \int_{|k|>\sigma} \D k\,\Big|  \sum_{j=1}^N  \frac{ e_j \hat \varphi (k) }{|k|^{\frac{1}{2}}} \ee^{-\im k \cdot x_j}  \left( \tfrac{ 1 }{|k|}\langle \kappa, \ddot x_j \rangle+ \langle  \kappa,\dot x_j \rangle^2 \right) \Big|^2 = \Or(\ln(\sigma^{-1}))\,,
\end{equation}
 the  claims follow from Lemma~\ref{philem}.
  \end{proof}
  
\begin{prop}
\label{prop13}
For $\sigma,\epsi\in(0,\frac13]$ and $M\in\N$  we have  that 
\[
\norm{\Big(U^{ \sigma}_\mathrm{dress}(t)-\ee^{-\frac{\im}{\epsi}( H_{\mathrm f} t +\int_0^t E_\sigma^\epsi(s)\,\D s)}\Big)Q_{\leq M}}=\mathcal{O}\Big(\epsi\,\sqrt{M+1}\,  {\ln(\sigma^{-1})}\Big)\,
\]
uniformly on bounded time intervals in $\mathcal{L}(\fock)$ and $\mathcal{L}(D(H_\mathrm{f}))$.

\end{prop}

\begin{proof}
Note that (\ref{UDNorm}) together with Lemma~\ref{lem:V}
implies $U^{ \sigma}_\mathrm{dress}( s)\in \mathcal{L}(D(H_\mathrm{f}))$ uniformly for $s\in[0,t]$.
With the shorthand $U_0(t):=\ee^{-\frac{\im}{\epsi}( H_{\mathrm f} t +\int_0^t  E_\sigma^\epsi(s)\,\D s)}$ and the previous lemma, we thus find
\begin{eqnarray*}
U^{ \sigma}_\mathrm{dress}(t)Q_{\leq M}&=& U_0(t)Q_{\leq M} +\frac{\im}{\epsi} U^{ \sigma}_\mathrm{dress}(t)\int_0^t \D s\, {U^{ \sigma}_\mathrm{dress}}^*(s)\,\epsi^2\, \dot \Phi_2(t) \,U_0(t)\,Q_{\leq M}\\
&=&U_0(t)Q_{\leq M} +\im\epsi\, U^{ \sigma}_\mathrm{dress}(t)\int_0^t \D s\, {U^{ \sigma}_\mathrm{dress}}^*(s) \,\dot \Phi_2(t) \, Q_{\leq M} \,U_0(t)\\
&=&U_0(t)Q_{\leq M}+\mathcal{O}\Big(\epsi  \,{\ln(\sigma^{-1})}\Big) \,.   \qedhere
\end{eqnarray*}
\end{proof}
For the next proposition  we abbreviate
\[
 U^\sigma_{\rmm{eff}}(t):=\ee^{-\frac{\im}{\epsi}( H_{\mathrm f} t +\int_0^t E_\sigma^\epsi(s)\,\D s)} \left(1+ \im \epsi 
\int_0^t \D s\, \ee^{\im H_{\mathrm f} \frac{s}{\epsi}} h^\sigma_\mathrm{rad}(s) \ee^{-\im H_{\mathrm f} \frac{s}{\epsi}}\right)
\]
with
\[
h^\sigma_\mathrm{rad}(t):=\Phi\Big( \sum_{j=1}^N   \frac{ e_j \hat \varphi_\sigma(k) }{|k|^{\frac{3}{2}}}\ee^{-\im k \cdot x_j(t)}  \langle  \kappa, \ddot x_j(t)\rangle \Big)\,.
\]
\begin{prop}\label{lem:letzte}
For   $ \sigma,\epsi\in(0,\frac13]$ and $M\in\N$ we have  that 
\begin{align}\label{UU2}
\norm{\Big(U^{ \sigma}_\mathrm{dress}(t)
-U^\sigma_{\rmm{eff}}(t)\Big) Q_{\leq M} } 
= \;\mathcal{O}\Big( \epsi^2 \,\sqrt{M+1}\, \ln \big(\sigma^{-1}\big)\Big)\,
\end{align}
uniformly on bounded time intervals  in $\mathcal{L}(\fock)$ and $\mathcal{L}(D(H_\mathrm{f}))$.
\end{prop}
\begin{proof}
Using again the Duhamel argument together with Proposition~\ref{prop13} we find
\begin{align*}
U^{ \sigma}_\mathrm{dress}(t)Q_{\leq M}&= U_0(t)Q_{\leq M} +\frac{\im}{\epsi} U^{ \sigma}_\mathrm{dress}(t)\int_0^t \D s\, {U^{ \sigma}_\mathrm{dress}}^*(s)\,\epsi^2 \,\dot \Phi_2 (t)\,U_0(t)\,Q_{\leq M}\\
&=U_0(t)Q_{\leq M} +\im\epsi\, U^{ \sigma}_\mathrm{dress}(t)\int_0^t \D s\, {U^{ \sigma}_\mathrm{dress}}^*(s)\,  Q_{\leq M+1} \,\dot \Phi_2(t) \, Q_{\leq M}\, U_0(t)\\
&=U_0(t)Q_{\leq M} +\im \epsi\, U_0(t) \int_0^t \D s\, U_0^*(s)\, \dot \Phi_2 (t) \,Q_{\leq M} \,U_0(t)+\mathcal{O}\Big(\epsi^2 \ln(\sigma^{-1})\Big)
\end{align*}
in the norm of ${\mathcal{L}(D(H_\mathrm{f}))}$. 
Inspecting (\ref{v2dot}) shows that
\[
\dot \Phi_2 \;=\;  h^\sigma_\rmm{rad}+
\Phi\Big( \sum_{j=1}^N \frac{e_j \hat \varphi_\sigma(k) }{|k|^{\frac{1}{2}}} \ee^{-\im k \cdot x_j}  \langle \kappa, \dot x_j \rangle^2\Big) \;=:\; h^\sigma_\rmm{rad}+  \Phi(f_\sigma )\,.
\]
Hence (\ref{UU2}) follows, once we show that the second term does not contribute to the leading order transitions.
This follows from an integration by parts,
\begin{eqnarray*}\lefteqn{
 \int_0^t \D s\,\ee^{\im H_{\mathrm f} \frac{s}{\epsi}} \Phi\left(f_\sigma(s) \right) 
\ee^{-\im H_{\mathrm f} \frac{s}{\epsi}} Q_{\leq M} \;=\;
 \int_0^t \D s\,  \Phi\left( \E^{-\I|k|\frac{s}{\epsi}} f_\sigma(s) \right) 
  Q_{\leq M} }
\\
&=& \I\epsi\, \int_0^t \D s\,  \Phi\left( \tfrac{1}{|k|}\left(\tfrac{\D}{\D s}\E^{-\I|k|\frac{s}{\epsi}} \right)f_\sigma(s) \right) 
  Q_{\leq M} 
  \\
& =&  \I\epsi\,   \Phi\left( \tfrac{1}{|k|}\,\E^{-\I|k|\frac{s}{\epsi}} \,f_\sigma(s) \right) 
  Q_{\leq M} \Big|_{0}^t\;-\; \I\epsi\, \int_0^t \D s\,  \Phi\left( \tfrac{1}{|k|}\E^{-\I|k|\frac{s}{\epsi}} \left(\tfrac{\D}{\D s}f_\sigma(s)\right) \right) 
  Q_{\leq M} \\
  &=&   \mathcal{O}\Big( \epsi \,\sqrt{M+1}  \,{\ln \big(\sigma^{-1}\big)}\,\Big)  \, . \qedhere
\end{eqnarray*}
\end{proof}
Now Theorem~1 and 2 follow from Proposition~\ref{lem:letzte} and Proposition~\ref{prop:approxs}.
To see this put $\sigma(\epsi) := \epsi^8$ and define
\[
V^\epsi(t) := V^\epsi_{\sigma(\epsi)}(t)\quad\mbox{and}\quad U^\epsi_{\rmm{eff}}(t) := U^{\sigma(\epsi)}_{\rmm{eff}}(t)\,.
\]
Then
\begin{align*}
\Big\| \Big(U_\mathrm{dress}(t)
-U^\epsi_\mathrm{eff}(t)&\Big)  Q_{\leq M} \Big\|\;  \leq \\&\leq\; \norm{\Big(U_\mathrm{dress}(t)
-U^{\sigma(\epsi)}_\mathrm{dress}(t)\Big) Q_{\leq M} }+\norm{\Big(U^{\sigma(\epsi)}_\mathrm{dress}(t)
-U^\epsi_{\rmm{eff}}(t)\Big) Q_{\leq M} }\\[1mm]
&= \; \Or(\epsi^7\,\sqrt{M+1}) + \Or(\epsi^2 \ln(\epsi^{-1})\,\sqrt{M+1})\,.
\end{align*}

\subsection{Non-adiabatic transitions}

It is now evident how to define the part of the wave function that corresponds to emitted photons. As an immediate consequence of  Theorem~\ref{thm:foae} we obtain the following corollary. 
\begin{cor}\label{thm:cta}
Let  $\phi(0) \in Q_m D(H_\mathrm{f}) $ and $\sigma(\epsi)=\epsi^8$ then 
\[
\phi_\mathrm{su}(t):= Q_m U_\rmm{dress}(t,0) \phi(0)  \quad\mbox{ and } \quad
\phi_\mathrm{na}(t):= Q_m^\perp U_\rmm{dress}(t,0) \phi(0) 
\]
have the expansions
\[
\phi_{\mathrm{su}}(t)= \ee^{-\frac{\im}{\epsi}(H_{\mathrm{f}} t +\int_0^t  E^\epsi(s)\,\D s)} \phi(0)+R_1
\]
and
\[
\phi_\mathrm{na}(t)=\im \epsi \ee^{-\frac{\im}{\epsi}(H_{\mathrm f} t +\int_0^t  E(s)\,\D s)} \int_0^t \D s\, \ee^{\im H_{\mathrm f} \frac{s}{\epsi}} h_\rmm{rad}(s)   \ee^{-\im H_{\mathrm f} \frac{s}{\epsi}} \phi(0)+ R_2
\] 
with $\norm{R_1}_{D(H_{\rm f})}=\mathcal{O}\big( \epsi^2  \ln  (\epsi ^{-1} )\big)$ and $\norm{R_2}_{D(H_{\rm f})}=\mathcal{O}\big( \epsi^2  \ln  (\epsi ^{-1} )\big)$.
\end{cor}

The probability for emitting a photon until time $t$ is thus given by $\|\phi_{\rm na}(t)\|^2$. But this has  no simple asymptotics for $\epsi\to 0$ because of the infrared problem. However, for the radiated energy there is a simple asymptotic expression. 

Recall the definition of the energy of free photons in \eqref{equ:rad}. Note first that
  at any  time $t$     when $\dot x(t) = \ddot x(t)=0$, it holds that $V^\epsi (t) = V_{\sigma(\epsi)}(t)$ and $E^\epsi(t) = E(t)$ and thus with Corollary~\ref{cor:spc}
\begin{align*} \nonumber
E_{\rm rad}(\psi(t)) & = \langle V_{\sigma(\epsi)}(t) \psi(t) , H_{\rm f} V_{\sigma(\epsi)}(t)\psi(t)\rangle  =  \langle \psi(t) , H^{\sigma(\epsi)} \psi(t) \rangle -  E_{\sigma(\epsi)}(t)\\
&=  \langle \psi(t) , H   \psi(t) \rangle -  E  (t)+ \Or( \sigma(\epsi)^\frac12 ) = E_{\rm rad, stat}(\psi(t))  + \Or(\epsi^4)\,.
\end{align*}
In the third equality we used equation \eqref{ineq11} to bound $\norm{H-H^\sigma}_{\mathcal{L}(D(H_\rmm{f}),\fock)}$ by $C \sigma^{1/2}$ and~(\ref{Esigma}). This shows that whenever the energy of free photons is unambiguously defined, our definition using the superadiabatic approximation agrees with it. 
 
  We will now prove Theorem~\ref{thm:rad} by plugging  the superadiabatic approximation into the definition \eqref{equ:rad}.
Note that this computation  of the energy is the reason that we insisted on all our estimates being valid also in  $\mathcal L(D(H_\mathrm{f}))$.

\begin{proof}[Proof of Theorem~\ref{thm:rad}]
For \eqref{equ:rad1}  observe that with the splitting given in Corollary~\ref{thm:cta} we have that
\begin{align*}
E_{\rm rad}(\psi(t)) & = \langle   \phi(t) , H_{\rm f}  \phi(t)\rangle =  \langle   \phi_{\rm su}(t) + \phi_{\rm na}(t), H_{\rm f}  ( \phi_{\rm su}(t) + \phi_{\rm na}(t))\rangle\\
& = \langle   \phi_{\rm na}(t) , H_{\rm f}  \phi_{\rm na}(t)\rangle  = E_{\rm rad}^0(t) + 2{\rm Re} \langle R_2, H_{\rm f} \phi_{\rm na}(t)\rangle + \langle R_2,H_{\rm f}R_2\rangle\,,
\end{align*}
where $E_{\rm rad}^0(t)$ just stands for the explicit expression claimed in \eqref{equ:rad1}. For the last term 
$\langle R_2,H_{\rm f}R_2\rangle = \mathcal{O}\big( \epsi^4  \ln  (\epsi ^{-1} )^2\big)$ follows from $\norm{R_2}_{D(H_{\rm f})} =\mathcal{O}\big( \epsi^2  \ln  (\epsi ^{-1} )\big)$ and for the mixed term
\[
 | \langle R_2, H_\rmm{f} \phi_{\rmm{na}}(t)  \rangle  |\leq \norm{R_2}_\fock  \norm{ H_\rmm{f} \phi_{\rmm{na}}(t)}_\fock =\mathcal{O}\Big( \epsi^4  (\ln  (\epsi ^{-1} )^2)\Big) 
\]
 from integrating $ H_\rmm{f} \phi_{\rmm{na}}(t)$ by parts as in the proof of Proposition~\ref{prop:approxs}.

%

%
%

 In order to get also  \eqref{equ:rad2}, 
we first transform the integral \eqref{equ:rad1} with  $\tau'=\frac{s-s'}{\epsi}$, $\tau=\frac{s+s'}{2}$ and define $a(\tau):= \min(\frac{2\tau}{\epsi},\frac{2}{\epsi}(t-\tau))$. The resulting integral is
\[
I(t) := \frac{\epsi^3}{2} \sum_{i,j=1}^N {e_j e_i} \int_{S^2} \D\kappa \int_0^\infty \D r\,   |\hat \varphi_\sigma(r)|^2
\int_0^t \int_{-a(\tau)}^{a(\tau)}  \D \tau\, \D \tau'\, \ee^{\im r( \tau'-\beta_{ji}(\tau,\epsi \tau',\kappa))} \alpha_{ji}(\tau,\epsi \tau',\kappa)\,,
\]
where we abbreviated $k=\kappa r$, $\alpha_{ji}(\tau,\epsi \tau',\kappa):=   \kappa \cdot \ddot x_j(\tau+\tfrac{\epsi \tau'}{2}) \,\kappa\cdot \ddot x_i(\tau-\tfrac{\epsi \tau'}{2}) $ and
$\beta_{ji}(\tau,\epsi \tau',\kappa):= \kappa\cdot (x_j(\tau+\tfrac{\epsi \tau'}{2})-x_i(\tau-\tfrac{\epsi \tau'}{2}))$.
Note that $\alpha_{ji}$ and $\beta_{ji}$ are uniformly bounded on the domain of integration.
Thus we can replace  $|\hat \varphi_\sigma(r)|^2$ by $ |\hat \varphi (r)|^2$ with an error of order $\epsi^2 \sigma |t|^2 $, which is negligible. 

Observing that after summation the integrand is  symmetric with respect to change of sign in $r$ and $\tau'$, i.e.\ for $f_{ji}(r,\tau'):= |\hat \varphi(r)|^2 \ee^{\im r( \tau'-\beta_{ji}(\tau,\epsi \tau',\kappa))} \alpha_{ji}(\tau,\epsi \tau',\kappa)$  
it holds that $\sum_{i,j} f_{ji}(r,\tau')=\sum_{i,j} f_{ji}(-r,-\tau')$, we can extend the $r$-integration to all of $\R$
  at a cost of a factor $\tfrac{1}{2}$.
 Modulo higher order terms this leads to 
\begin{align*}
I(t)  &=\frac{\epsi^3}{4} \sum_{i,j=1}^N {e_j e_i} \int_{S^2} \D \kappa\,
\int_0^t \int_{-a(\tau)}^{a(\tau)}   \D \tau\, \D \tau'\, \alpha_{ji}(\tau,\epsi \tau',\kappa) \int_{-\infty}^\infty \D r \,   |\hat \varphi(r)|^2 \ee^{\im r( \tau'-\beta_{ji}(\tau,\epsi \tau',\kappa))}\\
&= \frac{\sqrt{2\pi} \epsi^3}{4} \sum_{i,j=1}^N {e_j e_i} \int_{S^2} \D \kappa \, 
\int_0^t \int_{-a(\tau)}^{a(\tau)}   \D \tau\, \D \tau'\, \alpha_{ji}(\tau,\epsi \tau',\kappa) (\mathscr{F}   |\hat \varphi|^2) (\tau'-\beta_{ji}(\tau,\epsi \tau',\kappa))\,,
\end{align*}
where $\mathscr{F}$ denotes the Fourier transformation. 
Since with $\varphi$ also $ (\mathscr{F}   |\hat \varphi|^2)$ is a Schwartz function and since $\alpha_{ji}$ and $\beta_{ji}$ are uniformly bounded on the domain of integration, we can shrink $a(\tau)$ to $\tau_0$
 with an error that is asymptotically smaller then any inverse power of $\tau_0$ in the region where $\tau_0<a(\tau)$.
The same replacement in the region where $\tau_0>a(\tau)$ leads to an error of order $\Or(\epsi^4\tau_0)$. Taking e.g.\ $\tau_0= \epsi^{-1/2} $ leads to negligible errors. 
%
Now we can Taylor expand the integrand as
\begin{align*}
  \ee^{\im r( \tau'-\beta_{ji}(\tau,\epsi \tau' ))} \alpha_{ji}(\tau,\epsi \tau') 
&=  \ee^{\im r( \tau'-\beta_{ji}(\tau,0 ) +{ \mathcal O}(\epsi \tau_0 ))} (\alpha_{ji}(\tau,0)+ { \mathcal O}(\epsi \tau_0 )) \\
&= \ee^{\im r( \tau'-\beta_{ji}(\tau,0 )) }  \alpha_{ji}(\tau,0) + { \mathcal O}((1+r)\epsi \tau_0 ))
\end{align*}
and the remaining leading order contribution is
\begin{align*}
 I(t)& = \frac{\epsi^3}{4} \sum_{i,j=1}^N {e_j e_i} \int_{S^2}\hspace{-5pt} \D \kappa \, 
\int_0^t \hspace{-5pt} \D \tau \,\kappa\cdot\ddot x_j(\tau) \,\kappa\cdot\ddot x_i(\tau)  \int_{-\tau_0}^{\tau_0} \hspace{-8pt} \D \tau' \int_{-\infty}^\infty \hspace{-8pt}\D r\, |\hat \varphi(r)|^2  \ee^{\im r\tau'} \ee^{-\im r \kappa \cdot
(x_j(\tau)-x_i(\tau))} \\
&= \frac{\epsi^3}{4} \sum_{i,j=1}^N {e_j e_i} \int_{S^2}\hspace{-5pt} \D \kappa \, 
\int_0^t \hspace{-5pt} \D \tau \,\kappa\cdot\ddot x_j(\tau) \,\kappa\cdot\ddot x_i(\tau)  \int_{-\tilde\tau_0}^{\tilde \tau_0} \hspace{-8pt} \D \tilde \tau  \int_{-\infty}^\infty \hspace{-8pt}\D r\, |\hat \varphi(r)|^2  \ee^{\im r\tilde\tau}
 \end{align*}
 with  $ \tilde{\tau}= \tau'-\kappa \cdot (x_j(\tau)-x_i(\tau))$. With the same argument as above we can replace $\tilde\tau_0$ by $\infty$ while making an negligible error and end up with
 \begin{eqnarray*}
 I(t)&  =& \frac{\sqrt{2\pi}\epsi^3}{4} \sum_{i,j=1}^N {e_j e_i} \int_{S^2}  \D \kappa \, 
\int_0^t  \D \tau \,\kappa\cdot\ddot x_j(\tau) \,\kappa\cdot\ddot x_i(\tau)  \int_{-\infty}^{\infty}  \D \tilde \tau \, (\mathscr{F}|\hat\varphi|^2)(\tilde \tau)\\
&= & \frac{2\pi \epsi^3}{4} |\hat\varphi(0)|^2 \sum_{i,j=1}^N {e_j e_i} \int_0^t  \D \tau \int_{S^2}  \D \kappa  
\,\kappa\cdot\ddot x_j(\tau) \,\kappa\cdot\ddot x_i(\tau)  \\
&= & \frac{2\pi \epsi^3}{4}  \frac{1}{(2\pi)^3} \frac{4\pi}{3} \sum_{i,j=1}^N {e_j e_i}   \, 
\int_0^t   \D \tau \, \ddot x_j(\tau)  \cdot\ddot x_i(\tau)\,.\qedhere
 \end{eqnarray*}
%
%
\end{proof}

\end{document}